\documentclass[12pt, draftclsnofoot, onecolumn]{IEEEtran}
\usepackage[utf8]{inputenc}
\usepackage{tabularx}
\usepackage{amsmath}
\usepackage{amsthm}
\usepackage{amssymb}
\usepackage{enumerate}
\usepackage{latexsym}
\usepackage{mathtools}
\usepackage{multirow}
\usepackage[ruled,linesnumbered]{algorithm2e}
\usepackage{longtable}
\usepackage{makecell}

\SetKwRepeat{Do}{do}{while}

\usepackage[unicode=true,
bookmarks=true,bookmarksnumbered=true,bookmarksopen=true,bookmarksopenlevel=1,
breaklinks=false,pdfborder={0 0 0},pdfborderstyle={},backref=false,colorlinks=true,linkcolor=blue,citecolor=blue,urlcolor=blue,bookmarks=false]
{hyperref}

\makeatletter
\theoremstyle{plain}

\newtheorem{theorem}{Theorem}
\newtheorem{corollary}{Corollary}
\newtheorem{proposition}{Proposition}
\newtheorem{lemma}{Lemma}
\newtheorem{definition}{Definition}

\newtheorem{example}{Example}
\theoremstyle{remark}
\newtheorem{remark}{Remark}

\usepackage{color, colortbl}
\usepackage[noadjust,sort,compress]{cite}

\@ifundefined{showcaptionsetup}{}{%
	\PassOptionsToPackage{caption=false}{subfig}}
\usepackage{subfig}

\usepackage[justification=centering]{caption}

\definecolor{Gray}{gray}{0.9}
\usepackage[first=0,last=9]{lcg}

\begin{document}
	
	\title{
		Generalized BCH Codes and Twisted Goppa Codes Attaining Their Designed Distances
	}
	
	\author{Yaqi Chen, Hao Chen, Cunsheng Ding, Huimin Lao, Chao Liu and Conghui Xie
		
		\thanks{

			Yaqi Chen is with the College of Cyber Security, Jinan University, Guangzhou, Guangdong Province, 510632, China. (e-mail: chenyq@stu.jnu.edu.cn).
			Hao Chen is with the College of Information Science and
			Technology, Jinan University, Guangzhou, Guangdong Province, 510632, China. (e-mail: haochen@jnu.edu.cn).
			Cunsheng Ding is with the Department of Computer Science and Engineering, The Hong Kong University of Science and Technology, Hong Kong, China. (e-mail: cding@ust.hk).
			Huimin Lao is with Strategic Centre for Research in Privacy-Preserving Technologies and Systems, Nanyang Technological University, Singapore. (e-mail: huimin.lao@ntu.edu.sg).
			Chao Liu is with the College of Information Science and
			Technology, Jinan University, Guangzhou, Guangdong Province, 510632, China. (chliuu@163.com).
			Conghui Xie is with the Hetao Institute of Mathematics and Interdisciplinary Sciences, Shenzhen, Guangdong, 518033, China. (e-mail: xiech@himis-sz.cn).
			
			The research of Hao Chen was supported by NSFC Grant 62032009. The research of C. Ding was supported by Hong Kong Research Grants Council under Grant No. 16301123. 
			The research of Huimin Lao was supported by the National Research Foundation, Singapore and Infocomm Media Development Authority under its Trust Tech Funding Initiative. Any opinions, findings and conclusions or recommendations expressed in this material are those of the author(s) and do not reflect the views of
			National Research Foundation, Singapore and Infocomm Media Development Authority.
			The research of Conghui Xie was supported by NSFC Grant 12526523. 
			(Corresponding author: Huimin Lao)
	}}
	
	\maketitle
	
	\begin{abstract}
		
		Determining the true minimum distance of an alternant code remains a notoriously difficult problem in coding theory. In this paper, we study the minimum distances of generalized BCH codes and twisted Goppa codes through their parity-check matrices.
		We first give a necessary and sufficient condition for an alternant code to attain its designed distance and apply it to generalized BCH codes. 
		As applications, we prove that broad classes of generalized BCH codes have minimum distances equal to their designed distances. These classes provide explicit infinite families rather than isolated examples.
		We characterize when a twisted Goppa code
		$\Gamma(L,g,\eta)$ with $\deg g=t$ satisfies
		$d(\Gamma(L,g,\eta))=t+1$, and derive structured classes and infinite families attaining this distance.
		
	\end{abstract}
	
	
	\section{Introduction}
	\subsection{Background and Motivation}
	
	Determining the true minimum distance is a fundamental problem for linear codes. 
	Even for BCH codes, which are among the most extensively studied linear codes, the designed distance is usually only a lower bound and proving that this bound is sharp often requires additional structural arguments. 
	Exact results on the minimum distance of linear codes are therefore known only for special families, or are obtained by nonconstructive arguments.

	Alternant codes form one of the most important classes of linear codes. They are subfield subcodes of generalized Reed--Solomon codes and contain many classical families, including BCH codes~\cite{BCH1960,H1959}, Goppa codes~\cite{Goppa1970,Goppa1971}, and generalized BCH codes \cite{CC1975}. 
	From the parity-check viewpoint, an alternant code is obtained by allowing both the support and the column multipliers in a BCH-type parity-check matrix to vary. 
	This extra freedom in the definition allows alternant codes to contain asymptotically good families meeting the Gilbert--Varshamov bound, in contrast to the situation of BCH codes.
	
	Throughout this paper, for a linear code $C$, we denote by $d(C)$ its minimum distance. A $q$-ary linear code with length $n$, dimension $k$, and minimum distance $d$ is denoted by $[n,k,d]_q$.
	Let $\boldsymbol x=(x_1,\ldots,x_n)$ be a tuple of distinct elements of $\mathbb F_{q^m}$, and let $\boldsymbol y=(y_1,\ldots,y_n)\in(\mathbb F_{q^m}^*)^n$. 
	An alternant code $\mathcal A_r(\boldsymbol x,\boldsymbol y)$ of degree $r$ over $\mathbb F_q$ is defined by 
	$$
	\mathcal A_r(\boldsymbol x,\boldsymbol y)
	=
	\left\{
	\boldsymbol c\in\mathbb F_q^n:
	\sum_{i=1}^n c_i y_i x_i^\ell=0,\ 0\le \ell\le r-1
	\right\}.
	$$
	Equivalently, it is defined by the parity-check matrix $\mathbf H=\left(y_i x_i^\ell\right)_{\substack{0\le \ell\le r-1\\ 1\le i\le n}}.$
	Since the $x_i$ are distinct and the $y_i$ are nonzero, any $r$ columns of $\mathbf H$ are linearly independent over $\mathbb F_{q^m}$. 
	Thus $\mathcal A_r(\boldsymbol x,\boldsymbol y)$ has designed distance $\delta=r+1$.
	Therefore, determining when alternant codes have minimum distances equal to the designed distances is a natural problem.
	Several recent works have made progress in this direction. 
	In \cite{BCH2026}, the authors studied the minimum distances of some families of BCH codes and constructed explicit minimum weight codewords for  narrow-sense BCH codes. 
	In \cite{GoppaBCH2026}, the authors presented a criterion for Goppa codes to attain their designed distance and applied it to several families of Goppa codes and BCH codes.

	Generalized BCH (GBCH) codes are a large structured subclass of alternant codes.
	They were introduced by Chien and Choy through the Mattson--Solomon transform as an algebraic generalization of BCH codes~\cite{CC1975}. 
	This class contains BCH codes and Goppa codes as subclasses, and it also unifies Helgert's noncyclic generalizations of BCH and Srivastava codes~\cite{Helgert1972}. 
	In this sense, GBCH codes should be viewed not merely as variants of BCH codes, but as a large structured subclass of alternant codes.
	
	The normalized form of GBCH codes makes this alternant structure particularly transparent. Let $\gcd(n,q)=1$ and $m=\operatorname{ord}_n(q)$, let $\alpha\in\mathbb F_{q^m}$ be a primitive $n$-th root of unity, and define
	$$
	\mu_n=\{x\in\mathbb F_{q^m}:x^n=1\}.
	$$
	After normalization, a GBCH code with designed distance $\delta$ can be viewed as an alternant code on the support $\mu_n$ with multiplier $xU(x)$, where the polynomial $U(X)$ is nonzero on $\mu_n$. 
	Write $\mu_n=\{x_0,x_1,\ldots,x_{n-1}\}$. Equivalently, its
	parity-check equations have the form
	$$
	\sum_{i=0}^{n-1}c_iU(x_i)x_i^j=0,
	\qquad 1\le j\le\delta-1.
	$$
	The support of $n$-th root of unity preserves a BCH-type algebraic structure, while the normalized multiplier $U(X)$ provides a much larger design space than the usual BCH defining sets.

	Twisted Goppa codes form a twisted analogue of classical Goppa codes. 
	They are subfield subcodes of twisted generalized Reed-Solomon codes, see~\cite{SuiYue2023}.
	And they are obtained from classical Goppa codes by replacing the last alternant row with a twisted row.  
	Let
	$g(x)\in\mathbb F_{q^m}[x]$ be monic and have degree $t$, let
	$L=\{\alpha_1,\ldots,\alpha_n\}\subseteq\mathbb F_{q^m}$, and assume that
	$g(\alpha_i)\ne0$ for all $i$. A parity-check matrix of the twisted Goppa code $\Gamma(L,g,\eta)$ has columns
	$$
	\frac{1}{g(\alpha_i)}
	\begin{pmatrix}
		1\\
		\alpha_i\\
		\vdots\\
		\alpha_i^{t-2}\\
		\alpha_i^{t-1}+\eta\alpha_i^t
	\end{pmatrix},
	\qquad 1\le i\le n.
	$$
	Thus the code keeps the Goppa multiplier $1/g(\alpha_i)$ and the usual alternant support structure, while the last row
	is replaced by the twisted row $\alpha_i^{t-1}+\eta\alpha_i^t$. 
	In general, we have $d(\Gamma(L,g,\eta))\ge t$. 
	Under suitable conditions, this lower bound improves to $d(\Gamma(L,g,\eta))\ge t+1$. In this paper, we study when the improved bound is sharp, and in this setting we refer to $t+1$ as the designed distance of $\Gamma(L,g,\eta)$.
	
	\subsection{Our Contributions} 
	
	The main contributions of this paper are as follows.
	
	\begin{itemize}
		\item We give a necessary and sufficient condition for an alternant code to attain its designed distance, see Theorem~\ref{T-0-alternant}.

		\item We apply Theorem~\ref{T-0-alternant} to normalized GBCH codes. 
		Viewing a GBCH code as an alternant code on the root-of-unity support $\mu_n$ with multiplier $xU(x)$, we obtain a necessary and sufficient condition for $d=\delta$ in terms of $U(X)$, see Theorem \ref{T-1-GBCH}.
		As applications, we prove that broad classes of GBCH codes determined by structured multipliers have true minimum distance equal to their designed distance, see Theorems~\ref{T-F-1}--\ref{T-F-3}.
		These classes contain explicit infinite families and examples not reducible to the ordinary BCH or Goppa cases on the same support.

		\item We prove a twisted analogue for twisted Goppa codes.
		For $\Gamma(L,g,\eta)$ with $\deg g=t$, we characterize when $d(\Gamma(L,g,\eta))=t+1$, see Theorem~\ref{T-Goppa-0}. We then derive two structured classes attaining this distance, see Theorems~\ref{T-Goppa-1} and~\ref{T-Goppa-2}. 
		In particular, these classes yield infinite families over fixed base fields with unbounded lengths.
		
	\end{itemize}

	\subsection{Organization} 
	
	The rest of this paper is organized as follows. 
	Section~\ref{s2} recalls the necessary background on alternant codes, GBCH codes, and twisted Goppa codes. 
	Section~\ref{s3} proves the criterion for alternant codes to attain the designed distance and applies it to normalized GBCH codes. 
	Section~\ref{s4} shows that broad classes of GBCH codes determined by structured multipliers have $d=\delta$. Section~\ref{s5} proves the criterion for twisted Goppa codes and gives structured classes with minimum distance equal to the designed distance. 
	Section~\ref{s6} concludes the paper.
	
	\section{Preliminaries}\label{s2}
	
	\subsection{Alternant codes and GBCH codes}	
	
	For a vector $\boldsymbol x=(x_1,\ldots,x_n)$, its entries are always assumed to be pairwise distinct whenever it is used as a support vector.
	Let $\boldsymbol x=(x_1,\ldots,x_n)\in\mathbb F_{q^m}^n$ be a support vector and let 
	$\boldsymbol y=(y_1,\ldots,y_n)\in(\mathbb F_{q^m}^*)^n$ be a multiplier vector.
	The generalized Reed--Solomon code is defined as follows.
	
	\begin{definition}[Generalized Reed--Solomon code]
		For $1\le k\le n$, the generalized Reed--Solomon code of dimension $k$ over $\mathbb F_{q^m}$ with support $\boldsymbol x$ and multiplier $\boldsymbol y$ is
		$$
		\operatorname{GRS}_k(\boldsymbol x,\boldsymbol y)
		=
		\{(y_1f(x_1),\ldots,y_nf(x_n)):
		f(z)\in\mathbb F_{q^m}[z],\ \deg f<k\}.
		$$
		It is an $[n,k,n-k+1]_{q^m}$ MDS code.
	\end{definition}
	
	For $1\le r\le n$, define
	\begin{equation}\label{eq:alternant-pcm}
		\mathbf H_{\mathcal A}(\boldsymbol x,\boldsymbol y)=
		\begin{pmatrix}
			y_1 & y_2 & \cdots & y_n\\
			y_1x_1 & y_2x_2 & \cdots & y_nx_n\\
			\vdots & \vdots & & \vdots\\
			y_1x_1^{r-1} & y_2x_2^{r-1} & \cdots & y_nx_n^{r-1}
		\end{pmatrix}.
	\end{equation}
	
	\begin{definition}[Alternant code]
		The alternant code of degree $r$ over $\mathbb F_q$ with support $\boldsymbol x$ and multiplier $\boldsymbol y$ is
		$$
		\mathcal A_r(\boldsymbol x,\boldsymbol y)
		=
		\{\boldsymbol c\in\mathbb F_q^n:
		\mathbf H_{\mathcal A}(\boldsymbol x,\boldsymbol y)\boldsymbol c^T=0\}.
		$$
		Equivalently,
		$
		\mathcal A_r(\boldsymbol x,\boldsymbol y)
		=
		\operatorname{GRS}_r(\boldsymbol x,\boldsymbol y)^\perp\cap\mathbb F_q^n.
		$
	\end{definition}
	
	Equivalently, the monomial basis $1,z,\ldots,z^{r-1}$ may be replaced by any basis
	$f_0,\ldots,f_{r-1}$ of $\mathbb F_{q^m}[z]_{<r}$, which gives the row-equivalent matrix
	$$
	\bigl(y_i f_j(x_i)\bigr)_{0\le j\le r-1,\ 1\le i\le n}.
	$$
	
	\begin{lemma}[{\cite{Helgert1974}}]\label{L-1}
		The alternant code $\mathcal A_r(\boldsymbol x,\boldsymbol y)$ has parameters $[n,k,d]_q$ satisfying
		$$
		k\ge n-mr,\qquad d\ge r+1.
		$$
		The lower bound $r+1$ is called the designed distance of $\mathcal A_r(\boldsymbol x,\boldsymbol y)$.
		The distance bound follows because any $r$ columns of $\mathbf H_{\mathcal A}(\boldsymbol x,\boldsymbol y)$ are linearly independent over $\mathbb F_{q^m}$. 
	\end{lemma}
	
	We next recall generalized BCH codes in~\cite{CC1975}. The construction is based on the Mattson--Solomon transform, which is the finite field Fourier transform with respect to a primitive root of unity. Let $n$ be a positive integer
	with $\gcd(n,q)=1$, let $m=\operatorname{ord}_n(q)$, and let
	$\alpha\in\mathbb F_{q^m}$ be a primitive $n$-th root of unity.
	Define
	$T_n=\{A(X)\in\mathbb F_{q^m}[X]:\deg A<n\},$
	and
	$T_n^*=\{A(X)\in T_n:\gcd(A(X),X^n-1)=1\}.$
	For $A(X)\in\mathbb F_{q^m}[X]$, let $[A(X)]_n$ denote the unique representative in $T_n$ congruent to $A(X)$ modulo $X^n-1$.

	\begin{definition}[Mattson--Solomon transform]
		For $a(x)=\sum_{i=0}^{n-1}a_ix^i\in T_n$, its Mattson--Solomon
		transform with respect to $\alpha$ is
		$$
		\Phi_{\alpha}(a(x))=A(X)=\sum_{j=0}^{n-1}A_jX^j,\qquad
		A_j=a(\alpha^j).
		$$
		The inverse transform is
		$$
		\Phi_{\alpha}^{-1}(A(X))=\sum_{i=0}^{n-1}a_ix^i,\qquad
		a_i=\frac1n A(\alpha^{-i}),
		$$
		where $1/n$ is well defined in $\mathbb F_q$ since $\gcd(n,q)=1$.
	\end{definition}

	\begin{definition}[Generalized BCH code]
		Let $P(X),G(X)\in T_n^*$. The $q$-ary generalized BCH code $C(P,G)$
		associated with the ordered pair $(P(X),G(X))$ is
		$$
		C(P,G)=
		\left\{
		a(x)\in\mathbb F_q[x]:\deg a<n,\ 
		[P(X)A(X)]_n\equiv0\pmod{G(X)},\
		A(X)=\Phi_\alpha(a(x))
		\right\}.
		$$
	\end{definition}

	Let $p(x)=\Phi_{\alpha}^{-1}(P(X))=\sum_{i=0}^{n-1}p_ix^i$ and $g(x)=\Phi_{\alpha}^{-1}(G(X))=\sum_{i=0}^{n-1}g_ix^i.$
	Since $P(X),G(X)\in T_n^*$, we have $p_i\ne0$ and $g_i\ne0$ for all $i$. 
	Let
	$$
	h_i=\frac{p_i}{g_i},\qquad x_i=\alpha^{-i},\qquad i=0,\ldots,n-1.
	$$
	If $r=\deg G(X)$, then $C(P,G)$ has the parity-check matrix	
	\begin{equation}\label{eq:gbch-pcm}
		\mathbf H_{\rm GBCH}(P,G)=
		\begin{pmatrix}
			h_0x_0 & h_1x_1 & \cdots & h_{n-1}x_{n-1}\\
			h_0x_0^2 & h_1x_1^2 & \cdots & h_{n-1}x_{n-1}^2\\
			\vdots & \vdots & & \vdots\\
			h_0x_0^r & h_1x_1^r & \cdots & h_{n-1}x_{n-1}^r
		\end{pmatrix}.
	\end{equation}
	Equivalently, this is an alternant parity-check matrix of degree $r$ with support
	$
	(x_0,\ldots,x_{n-1})=(1,\alpha^{-1},\ldots,\alpha^{-(n-1)})
	$
	and multiplier
	$
	(h_0x_0,\ldots,h_{n-1}x_{n-1}).
	$
	Consequently, $C(P,G)$ has designed distance $r+1$.

	\subsection{Twisted Goppa codes}	
	Goppa codes form one of the most important subclasses of alternant codes.

	\begin{definition}[Goppa code {\cite{Goppa1970}}]
		Let \(L=(\alpha_1,\ldots,\alpha_n)\in\mathbb F_{q^m}^n\) be a support vector, and let \(g(z)\in\mathbb F_{q^m}[z]\) be a polynomial of degree \(r\) such that \(g(\alpha_i)\ne0\) for \(i=1,\ldots,n\). The Goppa code with support \(L\) and Goppa polynomial \(g(z)\) is
		\[
		\Gamma(L,g)=\left\{\boldsymbol c=(c_1,\ldots,c_n)\in\mathbb F_q^n:\sum_{i=1}^n\frac{c_i}{z-\alpha_i}\equiv0\pmod{g(z)}\right\}.
		\]
	\end{definition}

	Twisted Goppa codes were introduced by Sui and Yue as a twisted analogue of classical Goppa codes.

	\begin{definition}[Twisted Goppa code {\normalfont\cite[Definition~2.1]{SuiYue2023}}]\label{D-TGoppa}
		Let $g(z)\in\mathbb F_{q^m}[z]$ be a monic polynomial of degree $t$, let $L=(\alpha_1,\ldots,\alpha_n)$ be an ordered tuple of distinct elements of $\mathbb F_{q^m}$ such that $g(\alpha_i)\ne0$ for $1\le i\le n$, and let $\eta\in\mathbb F_{q^m}$. The twisted Goppa code with support $L$, Goppa polynomial $g(z)$, and twist parameter $\eta$ is 
		$$ 
		\Gamma(L,g,\eta)= \left\{ \boldsymbol c=(c_1,\ldots,c_n)\in\mathbb F_q^n: \sum_{i=1}^n c_i \left( \frac1{z-\alpha_i} -\frac{\eta\alpha_i^t}{g(\alpha_i)} \right) \equiv0\pmod{g(z)} \right\}. $$ 
	\end{definition}

	When $\eta=0$, Definition \ref{D-TGoppa} gives the classical Goppa code $\Gamma(L,g)$. In the following, we assume $\eta\ne0$ whenever $\eta^{-1}$ appears. By \cite[Proposition~2.2]{SuiYue2023}, $\Gamma(L,g,\eta)$ has a parity-check matrix $\mathbf H=\mathbf M\mathbf D$, where $$ \mathbf M= \begin{pmatrix} 1&\cdots&1\\ \alpha_1&\cdots&\alpha_n\\ \vdots&&\vdots\\ \alpha_1^{t-2}&\cdots&\alpha_n^{t-2}\\ \alpha_1^{t-1}+\eta\alpha_1^t&\cdots& \alpha_n^{t-1}+\eta\alpha_n^t \end{pmatrix} $$ and $$ \mathbf D= \operatorname{diag}\left( \frac1{g(\alpha_1)},\ldots,\frac1{g(\alpha_n)} \right). $$
	
	Since any $t-1$ columns of $\mathbf H$ contain a nonsingular Vandermonde submatrix, we have $d(\Gamma(L,g,\eta)) \ge t$. 
	The next lemma gives a simple condition when this bound improves to $t+1$.
	A proof is included for completeness.
	
	\begin{lemma}[{\normalfont\cite[Theorem~2.5]{SuiYue2023}}]\label{L-2}
		Let $V$ be an additive subgroup of $\mathbb F_{q^m}$. Suppose that
		$L\subseteq V$ and $\eta^{-1}\notin V.$
		Then every $t$ columns of $\mathbf H$ are linearly independent over $\mathbb F_{q^m}$.
		Consequently, $d(\Gamma(L,g,\eta))\ge t+1$.
	\end{lemma}
	
	\begin{proof}
		Since $\mathbf D$ is nonsingular and diagonal, it suffices to consider the corresponding columns of $\mathbf M$. 
		Take pairwise distinct elements $\alpha_{i_1},\ldots,\alpha_{i_t}\in L$. 
		A standard twisted Vandermonde determinant calculation gives 
		$$ \det \begin{pmatrix} 1&\cdots&1\\ \alpha_{i_1}&\cdots&\alpha_{i_t}\\ \vdots&&\vdots\\ \alpha_{i_1}^{t-2}&\cdots&\alpha_{i_t}^{t-2}\\ \alpha_{i_1}^{t-1}+\eta\alpha_{i_1}^t&\cdots& \alpha_{i_t}^{t-1}+\eta\alpha_{i_t}^t \end{pmatrix} = \prod_{1\le u<v\le t}(\alpha_{i_v}-\alpha_{i_u}) \left(1+\eta\sum_{u=1}^t\alpha_{i_u}\right). $$ 
		The product is nonzero because the $\alpha_{i_u}$ are pairwise distinct. Also, since $L\subseteq V$, we have $\sum_{u=1}^t\alpha_{i_u}\in V$. If $1+\eta\sum_{u=1}^t\alpha_{i_u}=0$, then $-\eta^{-1}\in V$, contradicting $\eta^{-1}\notin V$. Hence the determinant is nonzero. 
		Thus every $t$ columns of $\mathbf H$ are linearly independent, and $d(\Gamma(L,g,\eta))\ge t+1$.
	\end{proof}
	
	\section{Characterization of GBCH Codes with $d=\delta$}\label{s3}
	
	In this section, we first give a general criterion for alternant codes to attain their designed distances
	and then apply it to normalized GBCH codes.

	\begin{theorem}\label{T-0-alternant}
		Let $\mathcal A_r(\boldsymbol x,\boldsymbol y)$ be the alternant code defined by the parity-check matrix \eqref{eq:alternant-pcm}, and let $\delta=r+1$. 
		Then $d(\mathcal A_r(\boldsymbol x,\boldsymbol y))=\delta$ if and only if there exist distinct indices $i_0,i_1,\ldots,i_r\in\{1,\ldots,n\}$ such that, for
		$
		I=\{i_0,i_1,\ldots,i_r\}
		$
		and
		$
		F_I(z)=\prod_{i\in I}(z-x_i),
		$
		the following hold
		\begin{equation}\label{eq:T-0}
			\frac{y_{i_0}F_I'(x_{i_0})}{y_{i_j}F_I'(x_{i_j})}\in\mathbb F_q^*,
			\qquad 1\le j\le r.
		\end{equation}
		In this case, $\mathcal A_r(\boldsymbol x,\boldsymbol y)$ contains a codeword $\boldsymbol c=(c_1,\ldots,c_{n})$ of weight $\delta$ given by
		$$
		c_{i_0}=1,\qquad
		c_{i_j}=\frac{y_{i_0}F_I'(x_{i_0})}{y_{i_j}F_I'(x_{i_j})},\quad 1\le j\le r,
		$$
		and $c_i=0$ for $i\notin I$.
	\end{theorem}
	
	\begin{proof}
		By Lemma~\ref{L-1}, we have $d(\mathcal A_r(\boldsymbol x,\boldsymbol y)) \ge \delta$. Then $d(\mathcal A_r(\boldsymbol x,\boldsymbol y))=\delta$ if and only if there exists a codeword of weight $\delta$.

		Fix a subset $I\subseteq\{1,\ldots,n\}$ with $|I|=\delta$. 
		A vector $\boldsymbol c$ supported on \(I\) is a codeword in $\mathcal A_r(\boldsymbol x,\boldsymbol y)$  if and only if
		$$
		\sum_{i\in I}y_ic_ix_i^j=0,\qquad j=0,1,\ldots,r-1.
		$$
		Since the elements $x_i$ with $i\in I$, are pairwise distinct and $|I|=r+1$, the matrix
		$
		(x_i^j)_{0\le j\le r-1,\ i\in I}
		$
		has rank $r$, so its right kernel is one-dimensional. For $0\le j\le r-1$, the Lagrange interpolation of $z^j$ on $\{x_i:i\in I\}$ gives
		$$
		z^j=\sum_{i\in I}x_i^j\frac{F_I(z)}{(z-x_i)F_I'(x_i)}.
		$$
		Comparing the coefficient of $z^r$ gives
		$$
		\sum_{i\in I}\frac{x_i^j}{F_I'(x_i)}=0,\qquad 0\le j\le r-1.
		$$
		Hence the right kernel is generated by
		$
		\left(\frac1{F_I'(x_i)}\right)_{i\in I}.
		$

		It follows that a vector $\boldsymbol c$ supported on $I$ is a codeword if and only if there exists $\lambda\in\mathbb F_{q^m}^*$ such that
		$$
		y_ic_i=\frac{\lambda}{F_I'(x_i)},\qquad i\in I.
		$$
		Equivalently,
		\begin{equation}\label{eq:T-1}
			c_i=\frac{\lambda}{y_iF_I'(x_i)},\qquad i\in I.
		\end{equation}

		Assume first that $d(\mathcal A_r(\boldsymbol x,\boldsymbol y))=\delta$. Then there is a codeword $\boldsymbol c\in\mathbb F_q^n$ of weight $\delta$. Let $I=\operatorname{supp}(\boldsymbol c)$ and fix $i_0\in I$. 
		Since $c_i\in\mathbb F_q^*$ for all $i\in I$, it follows from \eqref{eq:T-1} that
		$$ \frac{y_{i_0}F_I'(x_{i_0})}{y_iF_I'(x_i)} = \frac{c_i}{c_{i_0}} \in\mathbb F_q^*,\qquad i\in I. $$ Thus \eqref{eq:T-0} holds.

		Conversely, suppose that $I$ and $i_0\in I$ satisfy \eqref{eq:T-0}. Define $\boldsymbol c=(c_1,\ldots,c_n)$ by $$ c_i= \frac{y_{i_0}F_I'(x_{i_0})}{y_iF_I'(x_i)},\qquad i\in I, $$ and $c_i=0$ for $i\notin I$. Then $\boldsymbol c\in\mathbb F_q^n$ and $\operatorname{wt}(\boldsymbol c)=|I|=\delta$. 
		Moreover, for $0\le j\le r-1$, $$ \sum_{i=1}^ny_ic_ix_i^j = y_{i_0}F_I'(x_{i_0}) \sum_{i\in I}\frac{x_i^j}{F_I'(x_i)} =0. $$ 
		Hence $\boldsymbol c\in\mathcal A_r(\boldsymbol x,\boldsymbol y)$. Since $d(\mathcal A_r(\boldsymbol x,\boldsymbol y))\ge\delta$, we obtain $d(\mathcal A_r(\boldsymbol x,\boldsymbol y))=\delta$.
	\end{proof}
	
	We now apply Theorem~\ref{T-0-alternant} to generalized BCH codes. Write $ \mu_n=\{x_0,x_1,\ldots,x_{n-1}\}\subseteq\mathbb F_{q^m} $ for the set of all $n$-th roots of unity. We first recall a useful normalization. By Theorem~7 of Chien and Choy~\cite{CC1975}, every GBCH code $C(P,G)$ with $r=\deg G(X)$ can be written in the form
	$
	C(P,G)=C(P^*,X^r),
	$
	where
	$
	P^*(X)=[X^rP(X)G(X)^{-1}]_n.
	$
	Thus, a GBCH code can be studied through a single normalized multiplier. For later use, we record this normalization and include a short proof.

	\begin{proposition}\label{P-normalized-GBCH} 
		Let $P(X),G(X)\in T_n^*$ and let $r=\deg G(X)$. Define $ U(X)=[P(X)G(X)^{-1}]_n$. 
		Then $U(X)\in T_n^*$ and $ C(P,G)=C([X^rU(X)]_n,X^r).$ 
		Conversely, for any $U(X)\in T_n^*$, the normalized GBCH code $C_U:=C([X^rU(X)]_n,X^r) $ has parity-check matrix 
		\begin{equation}\label{eq:normalized-gbch-pcm} 
			\mathbf H_U= \bigl(U(x_i)x_i^j\bigr)_{1\le j\le r,\ 0\le i\le n-1}.
		\end{equation} 
		Equivalently, $C_U$ is the alternant code of degree $r$ with support $(x_0,\ldots,x_{n-1})$ and multiplier $ y_i=x_iU(x_i)$ for $0\le i\le n-1. $
	\end{proposition}

	\begin{proof}
		Since $P(X),G(X)\in T_n^*$, it follows that  $U(X)=[P(X)G(X)^{-1}]_n$ also belongs to $T_n^*$. 
		Set
		$
		P^*(X)=[X^rU(X)]_n
		$
		and
		$
		G^*(X)=X^r.
		$
		For each $0\le i\le n-1$, we have $x_i\ne0$ and
		$$
		\frac{P^*(x_i)}{G^*(x_i)}
		=
		\frac{x_i^rU(x_i)}{x_i^r}
		=
		U(x_i)
		=
		\frac{P(x_i)}{G(x_i)}.
		$$
		Thus the two pairs $(P,G)$ and $(P^*,G^*)$ have the same values of $P/G$ on the defining support. 
		Then they define the same GBCH code, namely
		$
		C(P,G)=C([X^rU(X)]_n,X^r).
		$
		
		The parity-check matrix of the GBCH code $([X^rU(X)]_n,X^r)$ is 
		$
		\mathbf H_U=
		\bigl(U(x_i)x_i^j\bigr)_{1\le j\le r,\ 0\le i\le n-1}.
		$
		Since
		$$
		U(x_i)x_i^j=x_iU(x_i)x_i^{j-1},\qquad 1\le j\le r,
		$$
		this is an alternant parity-check matrix of degree $r$ with support $(x_0,\ldots,x_{n-1})$ and multiplier $y_i=x_iU(x_i)$.
	\end{proof}

	We next describe the relation of GBCH codes with BCH codes and Goppa codes.
	The following proposition records the two standard special cases in normalized GBCH codes.
	
	\begin{proposition}\label{P-classical-cases} 
		Let $U(X)\in T_n^*$ and let $ C_U=C([X^rU(X)]_n,X^r) $ be a normalized GBCH code. 
		\begin{enumerate} 
			\item[(i)] If $U(X)\equiv \gamma X^b\pmod{X^n-1}$ for some $\gamma\in\mathbb F_{q^m}^*$ and some integer $b$, then $C_U$ is the BCH code with $r$ consecutive defining zeros $ \alpha^{-(b+1)},\alpha^{-(b+2)},\ldots,\alpha^{-(b+r)}.$
			\item[(ii)] Let $G_0(X)$ be the unique polynomial with $\deg G_0<n$ such that \begin{equation}\label{eq:goppa-candidate} XU(X)G_0(X)\equiv1\pmod{X^n-1}. \end{equation} 
			If $\deg G_0(X)=r$, then $C_U$ is the Goppa code on the support $(x_0,\ldots,x_{n-1})$ generated by $G_0(X)$.
		\end{enumerate} 
	\end{proposition}
	
	\begin{proof}
		
		For (i), assume that $ U(X)\equiv \gamma X^b\pmod{X^n-1} $ with $\gamma\in\mathbb F_{q^m}^*$. Then $U(x_i)=\gamma x_i^b$ for all $i$, and hence $$ \mathbf H_U= \gamma\bigl(x_i^{b+j}\bigr)_{1\le j\le r,\ 0\le i\le n-1}. $$ The nonzero scalar $\gamma$ does not change the null space. Since $x_i=\alpha^{-i}$, this is the parity-check matrix of the BCH code with consecutive defining zeros $ \alpha^{-(b+1)},\alpha^{-(b+2)},\ldots,\alpha^{-(b+r)}. $ This proves (i).
		
		For (ii), since $U(X)\in T_n^*$ and $X$ is a unit modulo $X^n-1$, the polynomial $XU(X)$ is also a unit modulo $X^n-1$. Hence there is a unique polynomial $G_0(X)$ with $\deg G_0<n$ satisfying \eqref{eq:goppa-candidate}. Evaluating \eqref{eq:goppa-candidate} at $x_i$ gives $$ x_iU(x_i)=\frac1{G_0(x_i)},\qquad 0\le i\le n-1. $$ Thus $$ U(x_i)x_i^j = \frac{x_i^{j-1}}{G_0(x_i)},\qquad 1\le j\le r. $$ Therefore $$ \mathbf H_U = \left(\frac{x_i^\ell}{G_0(x_i)}\right)_{0\le \ell\le r-1,\ 0\le i\le n-1}. $$ This is the standard Goppa parity-check matrix on the support $(x_0,\ldots,x_{n-1})$ generated by $G_0(X)$ precisely when $\deg G_0(X)=r$. This proves (ii).
	\end{proof}

	Applying Theorem~\ref{T-0-alternant} to the normalized GBCH form in Proposition~\ref{P-normalized-GBCH} gives the following result. 
	For a general GBCH pair $(P,G)$, we first replace it by its normalized multiplier
	$
	U(X)=[P(X)G(X)^{-1}]_n.
	$

	\begin{theorem}\label{T-1-GBCH} 
		Let $U(X)\in T_n^*$ and let $ C_U=C([X^rU(X)]_n,X^r). $ Put $\delta=r+1$. Then $d(C_U)=\delta$ if and only if there exist distinct indices $i_0,i_1,\ldots,i_r\in\{0,\ldots,n-1\}$ such that, for $ I=\{i_0,i_1,\ldots,i_r\} $ and $ F_I(z)=\prod_{i\in I}(z-x_i), $  the following hold \begin{equation}\label{eq:gbch-exact-condition} 		
			\frac{x_{i_0}U(x_{i_0})F_I'(x_{i_0})} {x_{i_j}U(x_{i_j})F_I'(x_{i_j})} \in\mathbb F_q^*, \qquad 1\le j\le r. 
		\end{equation} 
		In this case, $C_U$ contains a codeword $\boldsymbol c=(c_0,\ldots,c_{n-1})$ of weight $\delta$ given by $$ c_{i_0}=1,\qquad c_{i_j}= \frac{x_{i_0}U(x_{i_0})F_I'(x_{i_0})} {x_{i_j}U(x_{i_j})F_I'(x_{i_j})}, \quad 1\le j\le r, $$ and $c_i=0$ for $i\notin I$. 
	\end{theorem}

	\begin{proof} 
		By Proposition~\ref{P-normalized-GBCH}, the code $C_U$ is the alternant code of degree $r$ with support $(x_0,\ldots,x_{n-1})$ and multiplier $ y_i=x_iU(x_i). $ The result follows immediately from Theorem~\ref{T-0-alternant}. 
	\end{proof}

	\section{GBCH Codes with $d=\delta$ from Structured Multipliers}\label{s4}
	
	In this section, we show that broad classes of GBCH codes have true
	minimum distances equal to their designed distances. We keep the notation
	of Section~\ref{s2}, let $\gcd(n,q)=1$, $m=\operatorname{ord}_n(q)$, and let 
	$\mu_n\subseteq\mathbb F_{q^m}$ be the set of all $n$-th roots of unity.
	More generally, for a divisor $s$ of $q^m-1$, write
	$
	\mu_s=\{\xi\in\mathbb F_{q^m}:\xi^s=1\}.
	$
	By Proposition~\ref{P-normalized-GBCH}, every GBCH code may be
	studied in the normalized form $C_U=C([X^rU(X)]_n,X^r)$, where
	$U(X)\in T_n^*$. We use Theorem~\ref{T-1-GBCH} to identify four general
	types of structured multipliers $U(X)$ for which the corresponding GBCH
	codes satisfy $d=\delta$.

	\subsection{Multipliers from $A(X^\delta)$}
	
	Throughout this subsection, let $\delta\mid n$ and $N=n/\delta$.
	We first consider multipliers of the form $U(X)=A(X^\delta)$, where $\gcd(A(Y),Y^N-1)=1$.
	
	\begin{theorem}\label{T-F-1} 
		Let $A(Y)\in\mathbb F_{q^m}[Y]$ satisfy $\deg A<N$ and $\gcd(A(Y),Y^N-1)=1$. 
		Define $ U(X)=A(X^\delta). $ 
		Then $U(X)\in T_n^*$, and the GBCH code $ C=C([X^{\delta-1}U(X)]_n,X^{\delta-1}) $ has minimum distance $d(C)=\delta$. 
	\end{theorem}

	\begin{proof} 
		Write $\mu_n=\{x_0,x_1,\ldots,x_{n-1}\}$. 
		Since $\deg A<N$, we have $\deg U<n$.
		Moreover, $\gcd(A(Y),Y^N-1)=1$ implies that $A(b)\ne0$ for all $b\in\mu_N$. 
		For each $0\le i\le n-1$, we have $x_i^\delta\in\mu_N$, and hence $ U(x_i)=A(x_i^\delta)\ne0. $ Thus $U(X)\in T_n^*$. Fix $b\in\mu_N$ and set $$ I_b=\{i\in\{0,\ldots,n-1\}:x_i^\delta=b\}. $$ 
		The map $\xi\mapsto \xi^\delta$ from $\mu_n$ onto $\mu_N$ has kernel of size $\delta$, so $|I_b|=\delta$. 
		Moreover, $$ F_{I_b}(z)=\prod_{i\in I_b}(z-x_i)=z^\delta-b. $$ Hence $F_{I_b}'(x_i)=\delta x_i^{\delta-1}$ for $i\in I_b$. 
		Since $\delta\mid n$ and $\gcd(n,q)=1$, the element $\delta$ is nonzero in $\mathbb F_{q^m}$. 
		Therefore, for every $i\in I_b$, $$ x_iU(x_i)F_{I_b}'(x_i) = x_iA(x_i^\delta)\delta x_i^{\delta-1} = \delta bA(b), $$ which is nonzero and independent of $i$. 
		Write $I_b=\{i_0,i_1,\ldots,i_{\delta-1}\}$. 
		Then, for $1\le j\le \delta-1$, $$ \frac{x_{i_0}U(x_{i_0})F_{I_b}'(x_{i_0})} {x_{i_j}U(x_{i_j})F_{I_b}'(x_{i_j})} = 1\in\mathbb F_q^*. $$ Thus condition \eqref{eq:gbch-exact-condition} holds for the set $I_b$.
		By Theorem~\ref{T-1-GBCH}, we obtain $d(C)=\delta$.
	\end{proof}
	
	Theorem~\ref{T-F-1} also gives infinite families once $q$ and $\delta$
	are fixed.
	
	\begin{corollary}\label{C-1}
		Fix a prime power $q$ and an integer $\delta\ge2$ with
		$\gcd(\delta,q)=1$. Let $A(Y)\in\mathbb F_q[Y]$ be nonzero and
		satisfy $A(1)\ne0$. Then there are infinitely many integers $N$ such
		that the GBCH code
		$C=C([X^{\delta-1}A(X^\delta)]_n,X^{\delta-1})$ of length $n=\delta N$ satisfies
		$d(C)=\delta$.
	\end{corollary}

	\begin{proof}
		Since $A(Y)$ has only finitely many roots and $A(1)\ne0$, there are
		infinitely many integers $N$ with $\gcd(N,q)=1$ for which no root of
		$A(Y)$ lies in $\mu_N$. For these $N$, we have
		$\gcd(A(Y),Y^N-1)=1$, and the result follows from
		Theorem~\ref{T-F-1} by taking $n=\delta N$ and
		$U(X)=A(X^\delta)$.
	\end{proof}

	\begin{example}
		The following examples illustrate Corollary~\ref{C-1}
		by giving explicit infinite families and instances with small length.
		\begin{itemize}
			\item Let $(q,\delta)=(2,3)$ and
			$A(Y)=Y^3+Y^2+Y\in\mathbb F_2[Y]$. The nonzero roots of $A(Y)$
			have order $3$, so $\gcd(A(Y),Y^N-1)=1$ whenever $3\nmid N$.
			Thus every $N>3$ with $\gcd(N,6)=1$ gives a GBCH code of length
			$n=3N$ and minimum distance $3$. For $N=7$, we have
			$U(X)=A(X^3)=X^9+X^6+X^3$, and
			$$
			C=C([X^2U(X)]_{21},X^2)=C(X^{11}+X^8+X^5,X^2),
			$$
			which has parameters $[21,13,3]_2$. The corresponding best known
			code is $[21,13,4]_2$ according to~\cite{codetable}.
			
			\item Let $(q,\delta)=(5,6)$ and $A(Y)=Y+2\in\mathbb F_5[Y]$.
			The unique root of $A(Y)$ is $3$, which has order $4$ in
			$\mathbb F_5^*$. Hence $\gcd(A(Y),Y^N-1)=1$ whenever $4\nmid N$.
			Thus every $N>1$ with $\gcd(N,5)=1$ and $4\nmid N$ gives a GBCH
			code of length $n=6N$ and minimum distance $6$. For $N=2$, we
			have $U(X)=A(X^6)=X^6+2$, and
			$$
			C=C([X^5U(X)]_{12},X^5)=C(X^{11}+2X^5,X^5),
			$$
			which has parameters $[12,5,6]_5$ and is best known according
			to~\cite{codetable}.
		\end{itemize}
	\end{example}

	Table~\ref{tab:1} gives further examples obtained from
	Theorem~\ref{T-F-1}. All parameters were verified by SageMath. By
	Proposition~\ref{P-classical-cases}, none of these examples falls
	into the BCH case or the Goppa case on the same support.

	\begin{table}[htbp]
		\centering
		\caption{Codes in Theorem~\ref{T-F-1}}
		\label{tab:1}
		\begin{tabular}{|c|c|c|c|c|c|c|}
			\hline
			$q$ & $m$ & $n$ & $\delta$ & $U(X)$ & Parameters of $C$ & Best known parameters\\
			\hline
			\multirow{2}{*}{$3$}
			& $4$ & $20$ & $4$ & $X^4+1$ & $[20,11,4]_3$ & $[20,11,6]_3$ \\
			\cline{2-7}
			& $4$ & $40$ & $8$ & $X^8+1$ & $[40,17,8]_3$ & $[40,17,14]_3$\\
			\hline
			
			\multirow{2}{*}{$5$}
			& $2$ & $12$ & $4$ & $X^4+2$ & $[12,6,4]_5$ & $[12,6,6]_5$\\
			\cline{2-7}
			& $2$ & $24$ & $6$ & $X^{12}+2$ & $[24,16,6]_5$ & $[24,16,6]_5$\\
			\hline
			
			\multirow{2}{*}{$7$}
			& $3$ & $18$ & $6$ & $X^6+2$ & $[18,5,6]_7$ & $[18,5,12]_7$\\
			\cline{2-7}
			& $4$ & $30$ & $10$ & $X^{10}+2$ & $[30,9,10]_7$ & $[30,9,17]_7$\\
			\hline
		\end{tabular}
	\end{table}
	
	\subsection{Multipliers from unions of power-map fibers}

	Throughout this subsection, let $\ell\mid n$ and put $N=n/\ell$.
	We next consider unions of fibers of the power map $x\mapsto x^\ell$ on
	$\mu_n$.
	
	\begin{theorem}\label{T-F-4}
		Let $S\subseteq\mathbb F_q^*\cap\mu_N$ with $|S|=t>0$, and put
		$\delta=t\ell$. Define $Q(Y)=\prod_{s\in S}(Y-s)$. Let
		$B(Y)\in\mathbb F_{q^m}[Y]$ satisfy
		$\gcd(B(Q(X^\ell)),X^n-1)=1$, and define
		$$
		U(X)=[X^{n-\ell}B(Q(X^\ell))]_n.
		$$
		Then $U(X)\in T_n^*$, and the GBCH code
		$C=C([X^{\delta-1}U(X)]_n,X^{\delta-1})$ has minimum distance
		$d(C)=\delta$.
	\end{theorem}
	
	\begin{proof}
		Since $X^n-1$ splits over $\mu_n$ with distinct roots, the condition
		$
		\gcd(B(Q(X^\ell)),X^n-1)=1
		$
		is equivalent to
		\begin{equation}\label{e-T4-1}
			B(Q(x^\ell))\ne0,\qquad x\in\mu_n.
		\end{equation}
		For $x\in\mu_n$, we have $x^{n-\ell}=x^{-\ell}$. Hence
		$$
		U(x)=x^{-\ell}B(Q(x^\ell)),\qquad x\in\mu_n.
		$$
		By \eqref{e-T4-1}, $U(x)\ne0$ for all $x\in\mu_n$, and therefore $U(X)\in T_n^*$.
		
		Set
		$
		I=\{i\in\{0,\ldots,n-1\}:x_i^\ell\in S\}.
		$
		The map $x\mapsto x^\ell$ maps $\mu_n$ onto $\mu_N$ and has kernel of size $\ell$. Thus each equation $x^\ell=s$ with $s\in S$ has exactly $\ell$ solutions in $\mu_n$, and hence
		$
		|I|=t\ell=\delta.
		$
		Moreover,
		$$
		F_I(z)=\prod_{i\in I}(z-x_i)
		=
		\prod_{s\in S}(z^\ell-s)
		=
		Q(z^\ell).
		$$
		
		For $i\in I$, let $s_i=x_i^\ell\in S$. Then
		\begin{equation}\label{e-T4-2}
			F_I'(x_i)=\ell x_i^{\ell-1}Q'(s_i).
		\end{equation}
		
		Also $Q(s_i)=0$, so
		\begin{equation}\label{e-T4-3}
			U(x_i)=x_i^{-\ell}B(0),\qquad i\in I.
		\end{equation}
		Since $|S| \ne 0 $, choose $s\in S$ and $x\in\mu_n$ with $x^\ell=s$. Then $Q(x^\ell)=Q(s)=0$, and \eqref{e-T4-1} implies
		$
		B(0)\ne0.
		$
		Moreover, $\ell\ne0$ in $\mathbb F_{q^m}$ because $\ell\mid n$ and $\gcd(n,q)=1$. 
		Combining \eqref{e-T4-2} and \eqref{e-T4-3}, we obtain
		$$
		x_iU(x_i)F_I'(x_i)
		=
		\ell B(0)Q'(s_i),\qquad i\in I.
		$$
		
		As $S\subseteq\mathbb F_q^*$ and $Q(Y)=\prod_{s\in S}(Y-s)$ has distinct roots, we have
		$
		Q'(s_i)\in\mathbb F_q^*
		$
		for every $i\in I$. Write
		$
		I=\{i_0,i_1,\ldots,i_{\delta-1}\}.
		$
		Then, for $1\le j\le \delta-1$,
		$$
		\frac{x_{i_0}U(x_{i_0})F_I'(x_{i_0})}
		{x_{i_j}U(x_{i_j})F_I'(x_{i_j})}
		=
		\frac{Q'(s_{i_0})}{Q'(s_{i_j})}
		\in\mathbb F_q^*.
		$$
		Therefore conditions \eqref{eq:gbch-exact-condition} hold for the set $I$. By Theorem~\ref{T-1-GBCH}, we obtain $d(C)=\delta$.
	\end{proof}

	\begin{example}
		The following choices give explicit infinite families from
		Theorem~\ref{T-F-4}.
		Further instances with small length are listed in
		Table~\ref{tab:3}.
		\begin{itemize}
			\item Let $(q,\ell)=(5,6)$, $S=\{1,-1\}$, and $B(Y)=2Y+1$.
			Then $t=2$, $\delta=12$, and $Q(Y)=Y^2-1$. Since the roots of
			$B(Q(Y))=2Y^2-1$ have order $8$, every even integer $N$ with
			$\gcd(N,5)=1$ and $8\nmid N$ gives a GBCH code of
			length $n=6N$ and minimum distance $12$.
			For $N=4$, we have $n=24$ and
			$U(X)=[X^{18}B(Q(X^6))]_{24}=2X^6-X^{18}$. Thus
			$$
			C=C([X^{11}U(X)]_{24},X^{11})=C(2X^{17}-X^5,X^{11})
			$$
			has parameters $[24,9,12]_5$ and is best known according
			to~\cite{codetable}.
			
			\item Let $(q,\ell)=(2,3)$, $S=\{1\}$, and $B(Y)=Y^3+1$.
			Then $t=1$, $\delta=3$, and $Q(Y)=Y+1$. Since the nonzero roots
			of $B(Q(Y))=(Y+1)^3+1$ have order $3$, every $N>1$ with
			$\gcd(N,6)=1$ gives a binary GBCH code of length
			$n=3N$ and minimum distance $3$.
			For $N=5$, we have $n=15$ and
			$U(X)=[X^{12}B(Q(X^3))]_{15}=1+X^3+X^6$. Thus
			$$
			C=C([X^2U(X)]_{15},X^2)=C(X^2+X^5+X^8,X^2),
			$$
			has parameters $[15,9,3]_2$. The corresponding best known
			code is $[15,9,4]_2$ according to~\cite{codetable}.
		\end{itemize}
	\end{example}

	\begin{table}[htbp]
		\centering
		\caption{Codes in Theorem~\ref{T-F-4}}
		\label{tab:3}
		\begin{tabular}{|c|c|c|c|c|c|c|c|}
			\hline
			$q$ & $m$ & $n$ & $\ell$ & $\delta$ & $U(X)$ & Parameters of $C$ & Best known parameters\\
			\hline
			
			$2$ & $4$ & $15$ & $3$ & $3$
			& $1+X^3+X^6$
			& $[15,9,3]_2$
			& $[15,9,4 ]_2$\\
			\hline
			
			\multirow{2}{*}{$3$}
			& $4$ & $20$ & $4$ & $4$
			& $2+2X^{16}$
			& $[20,12,4]_3$
			& $[20,12,4 ]_3$\\
			\cline{2-8}
			& $4$ & $40$ & $4$ & $4$
			& $1+X^4+2X^{36}$
			& $[40,28,4]_3$
			& $[40,28,6 ]_3$\\
			\hline
			
			\multirow{2}{*}{$4$}
			& $2$ & $15$ & $3$ & $3$
			& $1+X^3+X^{12}$
			& $[15,11,3]_4$
			& $[15,11,4 ]_4$\\
			\cline{2-8}
			& $4$ & $51$ & $3$ & $3$
			& $1+X^3+X^{48}$
			& $[51,43,3]_4$
			& $[51,43,5 ]_4$\\
			\hline
			
			\multirow{2}{*}{$5$}
			& $2$ & $12$ & $3$ & $6$
			& $X^3+2X^9$
			& $[12,5,6]_5$
			& $[12,5,6 ]_5$\\
			\cline{2-8}
			& $2$ & $24$ & $6$ & $12$
			& $X^6+2X^{18}$
			& $[24,9,12]_5$
			& $[24,9,12 ]_5$\\
			\hline
			
			\multirow{2}{*}{$7$}
			& $2$ & $24$ & $6$ & $12$
			& $X^6+2X^{18}$
			& $[24,7,12]_7$
			& $[24,7,14 ]_7$\\
			\cline{2-8}
			& $4$ & $60$ & $10$ & $20$
			& $X^{10}+X^{50}$
			& $[60,19,20]_7$
			& $[60,19,27 ]_7$\\
			\hline
			
			\multirow{2}{*}{$8$}
			& $2$ & $21$ & $3$ & $3$
			& $1+X^3+X^{18}$
			& $[21,17,3]_8$
			& $[21,17,4 ]_8$\\
			\cline{2-8}
			& $4$ & $35$ & $5$ & $5$
			& $1+X^5+X^{30}$
			& $[35,19,5]_8$
			& $[35,19,11 ]_8$\\
			\hline
			
			\multirow{2}{*}{$9$}
			& $2$ & $20$ & $4$ & $4$
			& $2+2X^{16}$
			& $[20,15,4]_9$
			& $[20,15,5 ]_9$\\
			\cline{2-8}
			& $2$ & $40$ & $4$ & $4$
			& $1+X^4+2X^{36}$
			& $[40,34,4]_9$
			& $[40,34,5 ]_9$\\
			\hline
		\end{tabular}
	\end{table}

	\subsection{Multipliers from linearized polynomials}
	
	Throughout this subsection, let $p$ be the characteristic of $\mathbb F_q$ and
	let $\delta=p^e$ with $e\ge1$. Let $V$ be an $e$-dimensional
	$\mathbb F_p$-subspace of $\mathbb F_{q^m}$, and define
	$L(X)=\prod_{v\in V}(X-v)$. For $a\in\mathbb F_{q^m}$, write
	$a+V=\{a+v:v\in V\}$. Then $L(X)$ is a $p$-linearized polynomial of
	degree $\delta$, and $L'(X)=\rho$ for some $\rho\in\mathbb F_{q^m}^*$.
	
	\begin{theorem}\label{T-F-2}
		Let $a\in\mathbb F_{q^m}$ satisfy $a+V\subseteq\mu_n$, where $n$ is a divisor of $q^m-1$. Let $B(Y)\in\mathbb F_{q^m}[Y]$ satisfy $\gcd(B(L(X-a)),X^n-1)=1$. Define
		$U(X)=[X^{n-1}B(L(X-a))]_n.$
		Then $U(X)\in T_n^*$, and the GBCH code
		$C=C([X^{\delta-1}U(X)]_n,X^{\delta-1})$
		has minimum distance $d(C)=\delta$.
	\end{theorem}

	\begin{proof}
		Since $X$ is a unit modulo $X^n-1$ and $\gcd(B(L(X-a)),X^n-1)=1$, we have
		$
		U(X)=[X^{n-1}B(L(X-a))]_n \in T_n^*
		$.
		Set
		$I=\{i\in\{0,\ldots,n-1\}:x_i\in a+V\}.$
		Since $a+V\subseteq\mu_n$ and $|V|=\delta$, we have $|I|=\delta$. Moreover,
		$$
		F_I(z)=\prod_{i\in I}(z-x_i)=L(z-a).
		$$
		Thus, for every $i\in I$,
		$$
		F_I'(x_i)=L'(x_i-a)=\rho.
		$$
		Also, since $L(x_i-a)=0$ and $x_i^n=1$, the definition of $U$ gives
		$$
		x_iU(x_i)=B(L(x_i-a))=B(0),\qquad i\in I.
		$$
		By assumption, $B(0)\ne0$, otherwise $a$ would be a common root of $B(L(X-a))$ and $X^n-1$. 
		Hence
		$$
		x_iU(x_i)F_I'(x_i)=B(0)\rho,\qquad i\in I,
		$$
		which is nonzero and independent of $i$.
		
		Write $I=\{i_0,i_1,\ldots,i_{\delta-1}\}$. Then, for $1\le j\le \delta-1$,
		$$
		\frac{x_{i_0}U(x_{i_0})F_I'(x_{i_0})}
		{x_{i_j}U(x_{i_j})F_I'(x_{i_j})}
		=
		1\in\mathbb F_q^*.
		$$
		Therefore conditions \eqref{eq:gbch-exact-condition} hold for the set $I$. By Theorem~\ref{T-1-GBCH}, we obtain $d(C)=\delta$.
	\end{proof}

	\begin{example}
		Theorem~\ref{T-F-2} gives the following infinite family. Let
		$n=q^m-1$ with $m\ge2$, take $V=\mathbb F_q$, and let
		$L(X)=X^q-X$. 
		Choose $a\in\mathbb F_{q^m}\setminus\mathbb F_q$ and
		$c\in\mathbb F_{q^m}$ with
		$\operatorname{Tr}_{\mathbb F_{q^m}/\mathbb F_q}(c)\ne0$, and $B(Y)=Y+a^q-a+c$. 
		Then
		$a+\mathbb F_q\subseteq\mathbb F_{q^m}^*=\mu_n$, and 
		$B(L(X-a))=X^q-X+c$ has no root in $\mathbb F_{q^m}$.
		Therefore Theorem~\ref{T-F-2} gives an infinite family of GBCH codes with length $n=q^m-1$ and minimum distance $q$.
		In this family,
		$$
		U(X)=[X^{n-1}(X^q-X+c)]_n=X^{q-1}-1+cX^{n-1}.
		$$
		For instance:
		\begin{itemize}
			\item Let $(q,m)=(7,2)$ and take $c=1$. Then $n=48$,
			$U(X)=X^6+6+X^{47}$, and
			$$
			C=C([X^6U(X)]_{48},X^6)=C(X^{12}+6X^6+X^5,X^6),
			$$
			which has parameters $[48,36,7]_7$. 
			The corresponding best known code has parameters $[48,36,8]_7$ according to~\cite{codetable}.
			
			\item Let $(q,m)=(5,2)$ and take $c=1$. Then $n=24$,
			$U(X)=X^4+4+X^{23}$, and
			$$
			C=C([X^4U(X)]_{24},X^4)=C(X^8+4X^4+X^3,X^4),
			$$
			which has parameters $[24,16,5]_5$. The corresponding best known
			code is $[24,16,6]_5$ according to~\cite{codetable}.
		\end{itemize}
	\end{example}
	
	Further instances with small length obtained from Theorem~\ref{T-F-2} are listed in Table~\ref{tab:2}.
	All parameters were verified by
	SageMath. By Proposition~\ref{P-classical-cases}, none of these
	examples falls into the BCH case or the Goppa case on the same
	support.
	
	\begin{table}[htbp]
		\centering
		\caption{Codes in Theorem~\ref{T-F-2}}
		\label{tab:2}
		\begin{tabular}{|c|c|c|c|c|c|c|}
			\hline
			$q$ & $m$ & $n$ & $\delta$ & $U(X)$ & Parameters of $C$ & Best known parameters\\
			\hline
			
			\multirow{2}{*}{$2$}
			& $5$ & $31$ & $2$
			& $X+1+X^{30}$
			& $[31,26,2]_2$
			& $[31,26,3 ]_2$\\
			\cline{2-7}
			& $7$ & $127$ & $2$
			& $X+1+X^{126}$
			& $[127,120,2]_2$
			& $[127,120,3 ]_2$\\
			\hline
			
			\multirow{2}{*}{$3$}
			& $2$ & $8$ & $3$
			& $X^2+2+X^7$
			& $[8,4,3]_3$
			& $[8,4,4 ]_3$\\
			\cline{2-7}
			& $4$ & $80$ & $3$
			& $X^2+2+X^{79}$
			& $[80,72,3]_3$
			& $[80,72,4 ]_3$\\
			\hline
			
			\multirow{2}{*}{$4$}
			& $3$ & $63$ & $4$
			& $X^3+1+X^{62}$
			& $[63,54,4]_4$
			& $[63,54,5 ]_4$\\
			\cline{2-7}
			& $4$ & $255$ & $4$
			& $X^3+1+cX^{254}$
			& $[255,243,4]_4$
			& $[255,243,5 ]_4$\\
			\hline
			
			\multirow{2}{*}{$5$}
			& $2$ & $24$ & $5$
			& $X^4+4+X^{23}$
			& $[24,16,5]_5$
			& $[24,16,6 ]_5$\\
			\cline{2-7}
			& $3$ & $124$ & $5$
			& $X^4+4+X^{123}$
			& $[124,112,5]_5$
			& $[124,112,6 ]_5$\\
			\hline
			
			$7$ & $2$ & $48$ & $7$
			& $X^6+6+X^{47}$
			& $[48,36,7]_7$
			& $[48,36,8 ]_7$\\
			\hline
			
			\multirow{2}{*}{$8$}
			& $2$ & $63$ & $8$
			& $X^7+1+cX^{62}$
			& $[63,49,8]_8$
			& $[63,49,9 ]_8$\\
			\cline{2-7}
			& $2$ & $63$ & $2$
			& $X+1+cX^{62}$
			& $[63,61,2]_8$
			& $[63,61,2 ]_8$\\
			\hline
			
			\multirow{2}{*}{$9$}
			& $2$ & $80$ & $9$
			& $X^8+8+X^{79}$
			& $[80,64,9]_9$
			& $[80,64,10 ]_9$\\
			\cline{2-7}
			& $2$ & $80$ & $3$
			& $X^2+2+X^{79}$
			& $[80,76,3]_9$
			& $[80,76,4 ]_9$\\
			\hline
		\end{tabular}
	\end{table}

	\subsection{Multipliers from irreducible polynomials}
	
	\begin{theorem}\label{T-F-3}
		Let $F(X)\in\mathbb F_q[X]$ be a monic irreducible divisor of $X^n-1$ with degree $\delta\ge2$. Let $\lambda\in\mathbb F_{q^m}^*$, and let $U(X)\in T_n^*$ satisfy
		\begin{equation}\label{e-T4-0} 		
			U(X)\equiv \lambda\bigl(XF'(X)\bigr)^{-1}\pmod{F(X)}.
		\end{equation} 
		Then the GBCH code
		$
		C=C([X^{\delta-1}U(X)]_n,X^{\delta-1})
		$
		has minimum distance $d(C)=\delta$. 
	\end{theorem}

	\begin{proof} 
		Since $\gcd(n,q)=1$, the polynomial $X^n-1$ is squarefree. Hence $F(X)$ is separable. Moreover, $F(0)\ne0$ because $F(X)\mid X^n-1$. 
		Thus $\gcd(F(X),XF'(X))=1$, so $XF'(X)$ is invertible modulo $F(X)$. 
		Define 
		$$ I=\{i\in\{0,\ldots,n-1\}:F(x_i)=0\}. $$ 
		Since $F(X)$ is monic of degree $\delta$ and splits into distinct roots in $\mu_n$, we have $|I|=\delta$ and $$ F_I(z)=\prod_{i\in I}(z-x_i)=F(z). $$ 
		
		For every $i\in I$, evaluating the congruence for $U(X)$ at $x_i$ gives $$ U(x_i)=\frac{\lambda}{x_iF'(x_i)}. $$ Since $F_I'(x_i)=F'(x_i)$, it follows that $$ x_iU(x_i)F_I'(x_i)=\lambda,\qquad i\in I. $$ 
		Write $I=\{i_0,i_1,\ldots,i_{\delta-1}\}$. Then, for $1\le j\le \delta-1$, $$ \frac{x_{i_0}U(x_{i_0})F_I'(x_{i_0})} {x_{i_j}U(x_{i_j})F_I'(x_{i_j})} = 1\in\mathbb F_q^*. $$ 
		Therefore conditions \eqref{eq:gbch-exact-condition} hold for the set $I$. By Theorem~\ref{T-1-GBCH}, we obtain $d(C)=\delta$.
	\end{proof}

	\begin{remark}\label{R1}
		Multipliers satisfying the congruence \eqref{e-T4-0} in Theorem~\ref{T-F-3} always exist. Write $ R(X)=(X^n-1)/F(X). $ Since $\gcd(F(X),R(X))=1$, the
		Chinese remainder theorem gives a unique class modulo $X^n-1$
		satisfying
		$$
		U(X)\equiv \lambda\bigl(XF'(X)\bigr)^{-1}\pmod{F(X)},
		\qquad
		U(X)\equiv1\pmod{R(X)}.
		$$
		The representative of this class in $T_n$ is nonzero on every point
		of $\mu_n$, and hence belongs to $T_n^*$.
	\end{remark}
	
	\begin{corollary}\label{C-2}
		Let $F(X)\in\mathbb F_q[X]$ be a monic irreducible polynomial of
		degree $\delta\ge2$ with $F(0)\ne0$. Let $\beta$ be a root of
		$F(X)$, and let $n=\operatorname{ord}(\beta)$. Then $F(X)\mid X^n-1$.
		Consequently, there exists a normalized GBCH code of length $n$ with
		minimum distance equal to its designed distance $\delta$.
	\end{corollary}
	
	\begin{proof}
		Since $F(0)\ne0$, the root $\beta$ is nonzero and has finite
		multiplicative order $n$. All conjugates
		$\beta,\beta^q,\ldots,\beta^{q^{\delta-1}}$ are also $n$-th roots of
		unity, so all roots of $F(X)$ are roots of $X^n-1$. Hence
		$F(X)\mid X^n-1$. The conclusion follows from
		Theorem~\ref{T-F-3} and Remark~\ref{R1}.
	\end{proof}

	\begin{example} 
		Let $(q,n,\delta)=(3,121,5)$. Then $\operatorname{ord}_{121}(3)=5$. Choose $ F(X)=X^5-X^4-X^2-1\in\mathbb F_3[X]. $ By SageMath calculation, $F(X)$ is an irreducible divisor of $X^{121}-1$. Moreover, $$ \bigl(XF'(X)\bigr)^{-1}\equiv -X^3-X\pmod{F(X)}. $$ Thus Theorem~\ref{T-F-3} applies with $ U(X)=-X^3-X. $ Since $U(X)$ has no zeros on $\mu_{121}$, we have $U(X)\in T_{121}^*$. 
		Hence the GBCH code $ C=C([X^4U(X)]_{121},X^4)=C(-X^7-X^5,X^4) $ has parameters $[121,102,5]_3$. 
		This code is neither in the BCH case nor in the Goppa case of Proposition~\ref{P-classical-cases}. 
	\end{example}

	\section{Twisted Goppa codes with $d=\delta$}\label{s5}
	
	In this section, we study when twisted Goppa codes attain the improved bound $\delta=t+1$.
	We first give a necessary and sufficient condition for $d(\Gamma(L,g,\eta))=t+1$, and then use this criterion to construct two structured classes attaining this distance.
	
	\begin{theorem}\label{T-Goppa-0}
		Let $\Gamma(L,g,\eta)$ be a twisted Goppa code with $\deg g=t$ and $L=\{\alpha_1 , \ldots, \alpha_n\}$. Let $V$ be an additive subgroup of $\mathbb F_{q^m}$ such that
		$
		L\subseteq V
		$
		and
		$
		\eta^{-1}\notin V.
		$
		Then $d(\Gamma(L,g,\eta))=t+1$ if and only if there exist distinct indices
		$i_0,i_1,\ldots,i_t\in\{1,\ldots,n\}$ such that, for
		$$
		F(z)=\prod_{\ell=0}^t(z-\alpha_{i_\ell}),
		\qquad
		s=\sum_{\ell=0}^t\alpha_{i_\ell},
		$$
		the following hold
		\begin{equation}\label{eq:T3-0} 		
			\frac{g(\alpha_{i_j})}{g(\alpha_{i_0})}
			\cdot
			\frac{F'(\alpha_{i_0})}{F'(\alpha_{i_j})}
			\cdot
			\frac{\eta^{-1}+s-\alpha_{i_j}}{\eta^{-1}+s-\alpha_{i_0}}
			\in\mathbb F_q^*,
			\qquad 1\le j\le t.
		\end{equation}

	\end{theorem}
	
	\begin{proof}
		By Lemma~\ref{L-2}, $d(\Gamma(L,g,\eta))\ge t+1$.
		Then $d(\Gamma(L,g,\eta))= t+1$ if and only if there exists a codeword of weight $t+1$.
		Fix distinct indices $i_0,i_1,\ldots,i_t$. A vector supported on these positions
		is in $\Gamma(L,g,\eta)$ if and only if
		$$
		\sum_{\ell=0}^t
		\frac{c_{i_\ell}}{g(\alpha_{i_\ell})}\alpha_{i_\ell}^r=0,
		\qquad 0\le r\le t-2,
		$$
		and
		$$
		\sum_{\ell=0}^t
		\frac{c_{i_\ell}}{g(\alpha_{i_\ell})}
		\left(\alpha_{i_\ell}^{t-1}+\eta\alpha_{i_\ell}^t\right)=0.
		$$
		This is the parity-check system restricted to the positions
		$i_0,i_1,\ldots,i_t$. Since every $t$ columns of the parity-check matrix are linearly independent,
		this system has a one-dimensional solution space.

		For $0\le \ell\le t$, define
		$$
		u_\ell=\frac{\eta^{-1}+s-\alpha_{i_\ell}}{F'(\alpha_{i_\ell})}.
		$$
		Then $u_\ell\ne0$ for all $\ell$. The denominator is nonzero because the
		$\alpha_{i_\ell}$ are distinct. If the numerator were zero, then
		$$
		-\eta^{-1}=s-\alpha_{i_\ell}
		=\sum_{k\ne \ell}\alpha_{i_k}\in V,
		$$
		contrary to $\eta^{-1}\notin V$.

		By Lagrange interpolation on the roots of $F$, for every polynomial
		$P(z)$ of degree at most $t$, the coefficient of $z^t$ in $P(z)$ is
		$$
		\sum_{\ell=0}^t
		\frac{P(\alpha_{i_\ell})}{F'(\alpha_{i_\ell})}.
		$$
		Taking $P(z)=z^r$ for $0\le r\le t$ gives
		$$
		\sum_{\ell=0}^t\frac{\alpha_{i_\ell}^r}{F'(\alpha_{i_\ell})}=0
		\quad (0\le r\le t-1),
		\qquad
		\sum_{\ell=0}^t\frac{\alpha_{i_\ell}^t}{F'(\alpha_{i_\ell})}=1.
		$$
		For $z^{t+1}$, we apply the same coefficient comparison to its remainder
		modulo $F(z)$. Since the coefficient of $z^t$ in $F(z)$ is $-s$, this
		remainder has coefficient $s$, and hence
		$$
		\sum_{\ell=0}^t
		\frac{\alpha_{i_\ell}^{t+1}}{F'(\alpha_{i_\ell})}=s.
		$$

		Therefore, for $0\le r\le t-2$,
		$$
		\sum_{\ell=0}^t u_\ell\alpha_{i_\ell}^r
		=
		(\eta^{-1}+s)
		\sum_{\ell=0}^t\frac{\alpha_{i_\ell}^r}{F'(\alpha_{i_\ell})}
		-
		\sum_{\ell=0}^t\frac{\alpha_{i_\ell}^{r+1}}{F'(\alpha_{i_\ell})}
		=0.
		$$
		Moreover,
		$$
		\begin{aligned}
			\sum_{\ell=0}^t
			u_\ell\left(\alpha_{i_\ell}^{t-1}
			+\eta\alpha_{i_\ell}^t\right)
			&=
			-\sum_{\ell=0}^t
			\frac{\alpha_{i_\ell}^{t}}{F'(\alpha_{i_\ell})}
			+
			\eta(\eta^{-1}+s)
			\sum_{\ell=0}^t
			\frac{\alpha_{i_\ell}^{t}}{F'(\alpha_{i_\ell})}
			-\eta
			\sum_{\ell=0}^t
			\frac{\alpha_{i_\ell}^{t+1}}{F'(\alpha_{i_\ell})}  \\
			&=
			-1+\eta(\eta^{-1}+s)-\eta s=0.
		\end{aligned}
		$$

		Thus $(u_0,\ldots,u_t)$ spans the solution space for the variables $c_{i_\ell}/g(\alpha_{i_\ell})$. Equivalently, every nonzero solution supported on these positions satisfies

		\begin{equation}\label{eq:T3-1}
			\frac{c_{i_j}}{c_{i_0}}
			=
			\frac{g(\alpha_{i_j})}{g(\alpha_{i_0})}
			\cdot
			\frac{u_j}{u_0},
			\qquad 1\le j\le t.
		\end{equation}
		
		Suppose first that $d(\Gamma(L,g,\eta))=t+1$. Let $\boldsymbol c$ be a codeword of weight $t+1$. After reordering the coordinates, we may assume that $ \operatorname{supp}(\boldsymbol c)=\{i_0,i_1,\ldots,i_t\}. $ Since $c_{i_0},c_{i_1},\ldots,c_{i_t}\in\mathbb F_q^*$, \eqref{eq:T3-1} gives $$ \frac{g(\alpha_{i_j})}{g(\alpha_{i_0})} \cdot \frac{F'(\alpha_{i_0})}{F'(\alpha_{i_j})} \cdot \frac{\eta^{-1}+s-\alpha_{i_j}}{\eta^{-1}+s-\alpha_{i_0}} = \frac{c_{i_j}}{c_{i_0}} \in\mathbb F_q^*, \qquad 1\le j\le t. $$ Hence the required conditions \eqref{eq:T3-0} hold.

		Conversely, suppose that distinct indices $i_0,i_1,\ldots,i_t$ satisfy \eqref{eq:T3-0}. Define a vector by setting $c_{i_0}=1$,
		$$
		c_{i_j}=
		\frac{g(\alpha_{i_j})}{g(\alpha_{i_0})}
		\cdot
		\frac{F'(\alpha_{i_0})}{F'(\alpha_{i_j})}
		\cdot
		\frac{\eta^{-1}+s-\alpha_{i_j}}{\eta^{-1}+s-\alpha_{i_0}},
		\qquad 1\le j\le t,
		$$
		and $c_i=0$ for all other positions. Then $\boldsymbol c\in\mathbb F_q^n$ and
		$\operatorname{wt}(\boldsymbol c)=t+1$. Moreover, the coordinate ratios above
		show that $\boldsymbol c$ satisfies the restricted parity-check system, so
		$\boldsymbol c\in\Gamma(L,g,\eta)$. This gives
		$d(\Gamma(L,g,\eta))=t+1$.
	\end{proof}

	We give the first class of twisted Goppa codes with $d=\delta$.
	
	\begin{theorem}\label{T-Goppa-1}
		Let $r,h,m$ be positive integers such that $r\mid m$, $r\le h<m$,
		and $q^r>2$. Let $L$ be an $h$-dimensional $\mathbb F_q$-subspace of
		$\mathbb F_{q^m}$ with $\mathbb F_{q^r}\subseteq L$. Choose
		$\eta\in\mathbb F_{q^m}^*$ such that $\eta^{-1}\notin L$, and define
		$$
		g(x)=(x-\eta^{-1})^{q^r-1}-1.
		$$
		Then $g(\alpha)\ne0$ for all $\alpha\in L$, and
		$d(\Gamma(L,g,\eta))=q^r$.
	\end{theorem}
	
	\begin{proof}
		Let $t=q^r-1$. We first show that $g$ has no zero on $L$. If
		$\alpha\in L$ and $g(\alpha)=0$, then
		$(\alpha-\eta^{-1})^{q^r-1}=1$. Since $r\mid m$, the roots of
		$z^{q^r}-z$ in $\mathbb F_{q^m}$ are exactly the elements of
		$\mathbb F_{q^r}$. Hence
		$\alpha-\eta^{-1}\in\mathbb F_{q^r}^*\subseteq L$, which contradicts
		$\eta^{-1}\notin L$.
		
		Since $L$ is additive and $\eta^{-1}\notin L$,
		Lemma~\ref{L-2} gives
		$d(\Gamma(L,g,\eta))\ge q^r$. It remains to identify a codeword of
		weight $q^r$. We apply Theorem~\ref{T-Goppa-0} to the support points
		in $\mathbb F_{q^r}$. For these points,
		$$
		F(z)=z^{q^r}-z,
		$$
		so $F'(z)=-1$ and their sum is zero.
		
		For every $\alpha\in\mathbb F_{q^r}$, we have
		$\alpha^{q^r}=\alpha$. Moreover, since
		$(\eta^{-1}-\alpha)=-(\alpha-\eta^{-1})$ and
		$(-1)^{q^r-1}=1$, we have
		$$
		(\eta^{-1}-\alpha)^{q^r-1}
		=
		(\alpha-\eta^{-1})^{q^r-1}.
		$$
		Therefore
		$$
		\eta^{-q^r}-\eta^{-1}
		=
		(\eta^{-1}-\alpha)^{q^r}-(\eta^{-1}-\alpha)
		=
		(\eta^{-1}-\alpha)g(\alpha),
		$$
		and hence
		$$
		g(\alpha)=\frac{\eta^{-q^r}-\eta^{-1}}{\eta^{-1}-\alpha},
		\qquad \alpha\in\mathbb F_{q^r}.
		$$
		
		Fix $\alpha_0\in\mathbb F_{q^r}$. For every
		$\alpha\in\mathbb F_{q^r}\setminus\{\alpha_0\}$, the condition \eqref{eq:T3-0} in
		Theorem~\ref{T-Goppa-0} becomes
		$$
		\frac{g(\alpha)}{g(\alpha_0)}
		\cdot
		\frac{F'(\alpha_0)}{F'(\alpha)}
		\cdot
		\frac{\eta^{-1}-\alpha}{\eta^{-1}-\alpha_0}
		=
		\frac{\eta^{-1}-\alpha_0}{\eta^{-1}-\alpha}
		\cdot
		\frac{\eta^{-1}-\alpha}{\eta^{-1}-\alpha_0}
		=1.
		$$
		By Theorem~\ref{T-Goppa-0}, we obtain $d(\Gamma(L,g,\eta))=q^r$.
	\end{proof}

	\begin{example}
		Let $(q,m,r,h)=(3,6,3,5)$, and let $\beta$ be a root of
		$X^6+X^5+X^4+1$ over $\mathbb F_3$. Set
		$$
		L=\mathbb F_{27}+\beta\mathbb F_9,\qquad
		\eta=\beta^{-2},\qquad
		g(x)=(x-\beta^2)^{26}-1.
		$$
		Then $L$ is a $5$-dimensional $\mathbb F_3$-subspace of
		$\mathbb F_{3^6}$, with $\mathbb F_{27}\subseteq L$, and
		$\eta^{-1}=\beta^2\notin L$. Hence Theorem~\ref{T-Goppa-1} gives
		$d(\Gamma(L,g,\eta))=27$. The code $\Gamma(L,g,\eta)$ has parameters
		$[243,101,27]_3$, where the dimension was verified by SageMath.
	\end{example}
	
	Theorem~\ref{T-Goppa-1} gives infinite families over any fixed base field.
	
	\begin{corollary}\label{C-TG-1}
		Fix a prime power $q$ and an integer $r\ge1$ with $q^r>2$. Then
		there exist infinitely many twisted Goppa codes over $\mathbb F_q$
		with unbounded lengths and minimum distance $q^r$.
	\end{corollary}
	
	\begin{proof}
		For every integer $h\ge r$, choose a multiple $m$ of $r$ with $m>h$.
		Then $\mathbb F_{q^r}\subseteq\mathbb F_{q^m}$, and there exists an
		$h$-dimensional $\mathbb F_q$-subspace
		$L\subseteq\mathbb F_{q^m}$ containing $\mathbb F_{q^r}$. Choose
		$\eta\in\mathbb F_{q^m}^*$ such that $\eta^{-1}\notin L$. By
		Theorem~\ref{T-Goppa-1},
		$$
		d(\Gamma(L,(x-\eta^{-1})^{q^r-1}-1,\eta))=q^r.
		$$
		The length is $|L|=q^h$, which is unbounded as $h\to\infty$.
	\end{proof}

	We give the second class of twisted Goppa codes with $d=\delta$.
	
	\begin{theorem}\label{T-Goppa-2}
		Let $r,m,t$ be positive integers such that $r\mid m$, $r<m$, and
		$t+1\mid q^r-1$. Let $S$ be the multiplicative subgroup of
		$\mathbb F_{q^r}^*$ of order $t+1$. Choose
		$\eta\in\mathbb F_{q^m}^*$ such that
		$\eta^{-1}\notin\mathbb F_{q^r}$, and define
		$$
		g(x)=x^t+\eta^t x^{t-1}+\cdots+\eta^2x+\eta.
		$$
		Let
		$
		L=\{\alpha\in\mathbb F_{q^r}:g(\alpha)\ne0\}.
		$
		Then $S\subseteq L$, and
		$d(\Gamma(L,g,\eta))=t+1$. Moreover, the length of
		$\Gamma(L,g,\eta)$ is at least $q^r-t$.
	\end{theorem}

	\begin{proof}
		The polynomial $g$ is monic of degree $t$. For every $a\in S$, using
		$a^{t+1}=1$, we have
		\begin{equation}\label{eq:TG2-0}
			a(\eta^{-1}-a)g(a)=\eta^{-1}-\eta^t.
		\end{equation}

		The right-hand side of \eqref{eq:TG2-0} is nonzero. Otherwise
		$\eta^{t+1}=1$, which would imply
		$\eta^{-1}\in S\subseteq\mathbb F_{q^r}$, a contradiction. Hence
		$g(a)\ne0$ for all $a\in S$, and so $S\subseteq L$.
		
		Since $L\subseteq\mathbb F_{q^r}$ and
		$\eta^{-1}\notin\mathbb F_{q^r}$,
		Lemma~\ref{L-2} gives
		$d(\Gamma(L,g,\eta))\ge t+1$. We now apply
		Theorem~\ref{T-Goppa-0} to the $t+1$ support points in $S$. For these
		points,
		$$
		F(z)=\prod_{a\in S}(z-a)=z^{t+1}-1.
		$$
		Thus $F'(a)=(t+1)a^{-1}$ for $a\in S$, and the sum of the elements
		of $S$ is zero. Fix $a_0\in S$. By \eqref{eq:TG2-0}, for every
		$a\in S\setminus\{a_0\}$,
		$$
		\frac{g(a)}{g(a_0)}
		=
		\frac{a_0(\eta^{-1}-a_0)}{a(\eta^{-1}-a)}.
		$$
		Therefore the condition \eqref{eq:T3-0} in Theorem~\ref{T-Goppa-0} becomes
		$$
		\frac{g(a)}{g(a_0)}
		\cdot
		\frac{F'(a_0)}{F'(a)}
		\cdot
		\frac{\eta^{-1}-a}{\eta^{-1}-a_0}
		=
		\frac{a_0(\eta^{-1}-a_0)}{a(\eta^{-1}-a)}
		\cdot
		\frac{a}{a_0}
		\cdot
		\frac{\eta^{-1}-a}{\eta^{-1}-a_0}
		=1.
		$$
		By Theorem~\ref{T-Goppa-0}, we obtain $d(\Gamma(L,g,\eta))=t+1$.	
		Finally, since $\deg g=t$, the polynomial $g$ has at most $t$ roots
		in $\mathbb F_{q^r}$. Thus $|L|\ge q^r-t$.
	\end{proof}

	\begin{example}
		The following examples are obtained from Theorem~\ref{T-Goppa-2}.
		In each item, for every listed value of $t$, let $S$ be the subgroup of
		$\mathbb F_{q^r}^*$ of order $t+1$, and define
		$$
		g(x)=x^t+\eta^t x^{t-1}+\cdots+\eta^2x+\eta.
		$$
		SageMath verified that $g$ has no zeros on the chosen support and computed the dimensions.
		
		\begin{itemize}
			\item Let $(q,r,m)=(2,8,16)$, and let $\beta$ be a primitive
			element of $\mathbb F_{2^{16}}$ with minimal polynomial
			$X^{16}+X^5+X^3+X^2+1$ over $\mathbb F_2$. Take
			$L=\mathbb F_{2^8}$ and $\eta=\beta^{-1}$. For $t=14,4,2$, we
			obtain twisted Goppa codes with parameters $[256,32,15]_2$, $[256,192,5]_2$, and $[256,224,3]_2$.
			
			\item Let $(q,r,m)=(3,6,12)$, and let $\beta$ be a primitive
			element of $\mathbb F_{3^{12}}$ with minimal polynomial
			$X^{12}+X^6+X^5+X^4+X^2+2$ over $\mathbb F_3$. Take
			$L=\mathbb F_{3^6}$ and $\eta=\beta^{-1}$. For $t=12,7,6$, we
			obtain twisted Goppa codes with parameters $[729,585,13]_3$, $[729,645,8]_3$, and $[729,657,7]_3$.
		\end{itemize}
	\end{example}
	
	Theorem~\ref{T-Goppa-2} also yields infinite families with fixed designed
	distance and unbounded length.
	
	\begin{corollary}\label{C-TG-2}
		Fix a prime power $q$ and an integer $\delta\ge2$ with
		$\gcd(\delta,q)=1$. Then there exist infinitely many twisted Goppa
		codes over $\mathbb F_q$ with unbounded lengths and minimum distance
		$\delta$.
	\end{corollary}
	
	\begin{proof}
		Let $e=\operatorname{ord}_\delta(q)$. For every integer $s\ge1$, let
		$r=es$ and $m=2r$. Then $\delta\mid q^r-1$, so
		$\mathbb F_{q^r}^*$ contains a multiplicative subgroup $S$ of order
		$\delta$. Choose $\eta\in\mathbb F_{q^m}^*$ such that
		$\eta^{-1}\notin\mathbb F_{q^r}$, and define
		$$g(x)=x^{\delta-1}+\eta^{\delta-1}x^{\delta-2}+\cdots+\eta^2x+\eta.$$
		Let
		$
		L=\{\alpha\in\mathbb F_{q^r}:g(\alpha)\ne0\}.
		$
		By Theorem~\ref{T-Goppa-2}, the code $\Gamma(L,g,\eta)$ has minimum
		distance $\delta$. Moreover, since $\deg g=\delta-1$, its length is
		at least $q^r-(\delta-1)$, which is unbounded as $s\to\infty$.
	\end{proof}

	{\section{Summary and concluding remarks}\label{s6}
		
		In this paper, we studied the minimum distances for GBCH codes and twisted Goppa codes. We first gave a necessary and sufficient condition for an alternant code to attain
		its designed distance. 
		For normalized GBCH codes, this condition was expressed in terms of the normalized multiplier $U(X)$.
		This led to broad classes of structured GBCH codes with minimum distances equal to their designed distances.
		
		We also proved a twisted analogue for twisted Goppa codes. More precisely, we characterized when a twisted
		Goppa code $\Gamma(L,g,\eta)$ with $\deg g=t$ satisfies
		$d(\Gamma(L,g,\eta))=t+1$. Applying this criterion, we constructed two structured classes of twisted Goppa codes attaining this distance, and obtained infinite families over fixed base fields with unbounded lengths.
		All examples in this paper were checked in SageMath by verifying $\mathbf H\boldsymbol c^\top =0$ for explicit codeword $\boldsymbol c$ of weight equal to the designed distance.
		
		Two related problems remain open. First, GBCH codes are known
		to be asymptotically good and can approach the Gilbert--Varshamov bound, but the known argument is existential. 
		It would be interesting to construct GBCH codes which are both asymptotically good and satisfy $d=\delta$. 
		Second, GBCH codes also have stronger BCH-type lower
		bounds, including the Hartmann--Tzeng bound, see~\cite[Corollary~5]{CC1975}.
		Determining when the GBCH codes have minimum distances attaining the Hartmann--Tzeng bound is naturally the next problem.}

\end{document}